\documentclass[12pt,oneside,english]{amsart}
\usepackage[T1]{fontenc}
\usepackage[latin9]{inputenc}
\usepackage{geometry}
\usepackage{amstext}
\usepackage{amsthm}
\usepackage{amssymb}

\makeatletter
\numberwithin{equation}{section}
\numberwithin{figure}{section}
\theoremstyle{plain}
\newtheorem{thm}{\protect\theoremname}
\theoremstyle{plain}
\newtheorem{lem}[thm]{\protect\lemmaname}
\theoremstyle{plain}
\newtheorem{prop}[thm]{\protect\propositionname}

\AtBeginDocument{
  
}

\makeatother

\usepackage{babel}
\providecommand{\lemmaname}{Lemma}
\providecommand{\propositionname}{Proposition}
\providecommand{\theoremname}{Theorem}

\begin{document}
\title{Pekar's Ansatz and the Ground-State Symmetry of a Bound Polaron}
\address{\noindent \hspace{-4ex}Georgia Institute of Technology, School of
Mathematics, Atlanta, Georgia 30332}
\address{\noindent \hspace{-4ex}rghanta3@math.gatech.edu}
\thanks{\noindent \hspace{-3.5ex}The paper is based on work supported by
the United States National Science Foundation Grant DMS-1600560. }
\author{Rohan Ghanta}
\begin{abstract}
\noindent We consider a Fröhlich polaron bound in a symmetric Mexican
hat-type potential. The ground state is unique and therefore invariant
under rotations. However, we show that the minimizers of the corresponding
Pekar problem are nonradial. Assuming these nonradial minimizers are
unique up to rotation, we prove in the strong-coupling limit that
the ground-state electron density converges in a weak sense to a rotational
average of the densities of the minimizers.
\end{abstract}

\maketitle

\section{Introduction}

In order to develop a theory of dielectric breakdown in semiconductors,
H. Fröhlich proposed a model in 1937 of an electron interacting with
the quantized optical modes (phonons) of an ionic crystal. Known today
as the (Fröhlich) polaron, it is one of the simplest examples of a
particle interacting with a quantized field, and perhaps most notably,
it has served as a testing ground for Feynman's path integral formulation
of quantum field theory. It is described by the Hamiltonian

\begin{equation}
H_{\alpha}^{V}=\mathbf{p}^{2}-\alpha^{2}V(\alpha x)+\int_{\mathbb{R}^{3}}a_{k}^{\dagger}a_{k}dk-\frac{\sqrt{\alpha}}{\left(2\pi\right)^{3/2}}\int_{\mathbb{R}^{3}}\left[a_{k}e^{ik\cdot x}+a_{k}^{\dagger}e^{-ik\cdot x}\right]\frac{dk}{|k|},\label{eq:1-3}
\end{equation}
acting on the Hilbert space $\mathcal{H}:=L^{2}(\mathbb{R}^{3})\otimes\mathcal{F}$
where $\mathcal{F}:=\oplus_{n\geq0}\otimes_{s}^{n}L^{2}(\mathbb{R}^{3})$
is the (symmetric) phonon Fock space. An outstanding idiosyncrasy
of the Fröhlich's polaron is that for all its popularity over the
years as a convenient ``toy model'' for a singularity-free field
theory, the electron-phonon interaction term in the Hamiltonian makes
it intractable for calculating even the most basic quantities such
as the effective mass and the ground-state energy. This computational
difficulty has led S.I. Pekar (in a series of collaborations with
L.D. Landau, O.F. Tomasevich and others between 1944 and 1950) to
derive from Fröhlich's model a much simpler\textendash{} albeit nonlinear\textendash{}
effective theory, built entirely on an (unjustified!) Ansatz for the
ground-state wave function. Remarkably, Pekar's effective minimization
problem nevertheless yields to leading order the exact ground-state
energy of the polaron in the strong-coupling limit $\alpha\rightarrow\infty$.
It is therefore natural to conjecture that the ground-state electron
density also converges (in a weak sense) to a minimizer of Pekar's
effective problem: after all, this is known to be the case for particular
one-dimensional models. In this paper, however, using an intuitive
example of a polaron localized in a radial potential, we shall showcase
a discrepancy in spherical symmetry between a rotation-invariant Hamiltonian
and its unique ground state on the one hand and the corresponding
Pekar Ansatz for the wave function on the other. This in turn illustrates
that such expected (weak) convergence of the ground state to Pekar's
minimizer is not in general true. 

We denote by $x\in\mathbb{R}^{3}$ the electron coordinate and by
$k\in\mathbb{R}^{3}$ the phonon mode; $\mathbf{p}=-i\nabla_{x}$
is the electron momentum, and the electric potential $V\in L^{3/2}(\mathbb{R}^{3})+L^{\infty}(\mathbb{R}^{3})$
is nonnegative and vanishes at infinity and usually arises from an
impurity in the crystal; $a_{k}^{\dagger}$ and $a_{k}$ are scalar
creation and annihlation operators on $\mathcal{F}$ which satisfy
the canonical commutation relation $[a_{k},a_{k'}^{\dagger}]=\delta(k-k')$;
and $\alpha>0$ is the electron-phonon coupling parameter. The \textit{ground-state
energy} of the model is defined to be
\begin{equation}
E^{V}\left(\alpha\right):=\inf\left\{ \left(\Psi,\ H_{\alpha}^{V}\Psi\right)_{_{\mathcal{H}}}\mid\left\Vert \Psi\right\Vert _{_{\mathcal{H}}}=1\right\} .\label{eq:1-2}
\end{equation}
Any normalized vector $\Omega\in\mathcal{H}$ that achieves the infimum
in (\ref{eq:1-2}) is called a \textit{ground-state wave function},
and it satisfies the Schrödinger equation $H_{\alpha}^{V}\Omega=E^{V}\left(\alpha\right)\Omega$;
integrating out its phonon coordinates, one has the \textit{electron
density }$\|\Omega\|_{\mathcal{F}}^{2}\left(x\right)$. Most of the
literature is concerned with the \textit{translation-invariant (TI-)
polaron}\textendash{} i.e., the case where $V\equiv0$ in (\ref{eq:1-3}).
It was shown in the 1980s that for all values of the coupling parameter
$\alpha>0$, the TI-polaron does not have a ground state (finally
settling a decades-long debate on the existence of a delocalization-localization
transition). We are instead interested in the case of nonzero $V$,
the \textit{bound polaron}, which has attracted sizable attention
(see {[}Dv1996{]} and the references therein; in particular, we refer
to the rigorous work on pinning transitions by H. Spohn {[}Sp1986{]}
and H. Löwen {[}Lw1988a{]}, {[}Lw1988b{]}). In contrast to the TI-polaron,
under physically natural conditions on hte external potential $V$,
the Fröhlich Hamiltonian $H_{\alpha}^{V}$ has a unique ground state
for all $\alpha>0$. This follows from now-standard techniques developed
by F. Hiroshima {[}Ha2000{]} and by M. Griesemer, E.H. Lieb and M.
Loss {[}GLL2001{]} to study the analogous Pauli-Fierz model in quantum
electrodynamics (see Appendix). Note that we have added the potential
in the scaled form $\alpha^{2}V(\alpha x)$ to the Hamiltonian in
order for its effect to survive in the limit $\alpha\rightarrow\infty$
(see Theorem 3.2 in {[}GW2013{]}). We work with the potential 

\begin{equation}
V_{R}\in C_{c}^{\infty}\left(\mathbb{R}^{3}\right),\ \ 0\leq V_{R}\leq1\ \ \text{and\ \ }V_{R}(x)=\left\{ \begin{array}{ccc}
0 & \text{when} & \ \ \ \ \,\,\left|x\right|\leq1\\
1 & \text{when} & 2\leq\left|x\right|\leq R\\
0 & \text{when} & \ \ \ \ \ \ \ \ \ \ \ \left|x\right|\geq R+1
\end{array}\right..\label{eq:P.1}
\end{equation}

First we motivate our results with a general potential. When the coupling
parameter $\alpha$ is large, Pekar guessed that the ground state
has the product form

\begin{equation}
\Psi_{\alpha}=\psi_{\alpha}(x)\otimes\Phi_{\alpha},\label{eq:1-1}
\end{equation}
where $\psi_{\alpha}\in L^{2}(\mathbb{R}^{3})$ is an electronic wave
function, and $\Phi_{\alpha}\in\mathcal{F}$ is a coherent state depending
only on the phonon coordinates: 
\begin{equation}
\Phi_{\alpha}=\prod_{k}\exp\left(z_{\alpha}(k)a_{k}^{\dagger}-\overline{z_{\alpha}(k)}a_{k}\right)\left|0\right\rangle \label{eq:a1}
\end{equation}
with the vacuum $\left|0\right\rangle \in\mathcal{F}$ and the phonon
displacements $z_{\alpha}(k)\in L^{2}(\mathbb{R}^{3})$, which are
to be determined variationally. In particular, $a_{k}\Phi_{\alpha}=z(k)\Phi_{\alpha}$. 

The optimization problem in (\ref{eq:1-2}) for the ground-state energy
becomes considerably more tractable if we assume that the ground state
has the product form in Pekar's Ansatz. Minimizing the quantity $\langle\Psi,\ H_{\alpha}^{V}\Psi\rangle$
over the more restrictive set of product wave functions in (\ref{eq:1-1})
and completing the square, Pekar deduced that
\begin{equation}
z_{\alpha}(k)=\frac{1}{\pi|k|}\sqrt{\frac{\alpha}{2}}\int_{\mathbb{R}^{3}}e^{-ik\cdot x}|\psi_{\alpha}|^{2}dx,\label{eq:a2}
\end{equation}
which in turn couples the coherent state to the electronic wave function
in (\ref{eq:1-1}), and arrived at an \textit{upper bound} for the
ground-state energy: 
\[
E_{\alpha}^{V}\leq\inf\left\{ \langle\Psi,\ H_{\alpha}^{V}\Psi\rangle\mid\|\Psi\|=1\ \mbox{and}\ \Psi=\psi\otimes\Phi\right\} 
\]
\begin{equation}
=\alpha^{2}e(V).\label{1-5}
\end{equation}
The quantity $e(V)$ in (\ref{1-5}) can be calculated by minimizing
the nonlinear \textit{Pekar functional}: 
\begin{equation}
e(V)=\inf_{\|\psi\|_{2}=1}\mathcal{E}_{V}(\psi),\label{eq:1-4.5}
\end{equation}
where 
\begin{equation}
\mathcal{E}_{V}(\psi)=\int_{\mathbb{R}^{3}}|\nabla\psi|^{2}dx-\int\int_{\mathbb{R}^{3}\times\mathbb{R}^{3}}\frac{\left|\psi(x)\right|^{2}\left|\psi(y)\right|^{2}}{\left|x-y\right|}dx\,dy-\int_{\mathbb{R}^{3}}V(x)|\psi(x)|^{2}dx.\label{eq:1-4}
\end{equation}
Furthermore, if the minimization problem problem in (\ref{eq:1-4.5})
admits a minimizer $\phi(x)$, then $\alpha^{3/2}\phi(\alpha x)$
is the electronic wave function in Pekar's product ground state from
(\ref{eq:1-1}): 

\begin{equation}
\Psi_{\alpha}=\alpha^{3/2}\phi(\alpha x)\prod_{k}\exp\left(z_{\alpha}(k)a_{k}^{\dagger}-\overline{z_{\alpha}(k)}a_{k}\right)\left|0\right\rangle ,\label{eq:1-4.25}
\end{equation}
where
\[
z_{\alpha}(k)=\frac{1}{\pi|k|}\sqrt{\frac{\alpha}{2}}\int_{\mathbb{R}^{3}}e^{-ik\cdot x}|\alpha^{3/2}\phi(\alpha x)|^{2}dx;
\]
note that the electronic function becomes more localized as the coupling
paramter $\alpha>0$ increases. 

Though Pekar's result in (\ref{1-5}) is only an upper bound, his
Ansatz provides the convenience of eliminating all of the phonon coordinates
from the calculation: the functional in (\ref{eq:1-4}) needs to be
minimized just over a single electronic coordinate, a sharp contrast
to the more demanding situation in (\ref{eq:1-2}). 

Not being amenable to the direct method in the calculus of variations,
Pekar's minimization problem for approximating the ground-state energy
in turn motivated mathematicians to develop novel and far-reaching
techniques in nonlinear analysis such as the symmetrization arguments
of E.H. Lieb, the Concentration-Compactness Lemma of P.L. Lions and
the stability theory of T. Cazenave and P.L. Lions. Indeed, the first
detailed analysis of the nonlinear problem in (\ref{eq:1-4}) was
given in 1977 by Lieb, who used rearrangement inequalities to show
that a minimizer exists when $V=0$. He also established that this
minimizer is unique up to a translation by proving uniqueness of a
radial solution for the corresponding Euler-Lagrange equation
\[
\left\{ -\triangle-2\int_{\mathbb{R}^{3}}|\phi(y)|^{2}|x-y|^{-1}dy\right\} \phi(x)=\phi(x),
\]
known in the literature as the \textit{Choquard-Pekar} or \textit{Schrödinger-Newton}
\textit{equation}. For showing the existence of a minimizer when $V\neq0$
in (\ref{eq:1-4}), Lieb's symmetrization argument applies for a symmetric
decreasing potential. This motivated Lions to develop his famous Concentration
Compactness Principle from 1984: for a general $V\geq0$ that vanishes
at infinity, he showed that the problem in (\ref{eq:1-4}) admits
a minimizer. Uniqueness of a minimizer when $V\neq0$, however, remains
an elusive open problem.

Despite giving rise to a rich variational theory that continues to
be a source of interesting mathematical problems, Pekar's Produkt-Ansatz
of the ground state in (\ref{eq:1-1}) lacks a rigorous justification:
It is based entirely on his \textit{feeling} that (we quote the amusing
yet accurate, anthropomorphic description from {[}LT1997{]}) ``...at
large coupling the phonons cannot follow the rapidly moving electron
(as they do at weak coupling) and so resign themselves to interacting
with the ``mean'' electron density $\psi^{2}(\mathbf{x})$.'' This
``mean-field'' interaction is reflected in the phonon displacements,
given in equation (\ref{eq:a2}), for Pekar's coherent state.) It
is therefore remarkable that Pekar's crude upper bound for the ground-state
energy in (\ref{1-5})\textendash{} derived after all from his unjustified
Ansatz\textendash{} becomes exact (to the leading order) in the strong-coupling
limit: 

\begin{equation}
\lim_{\alpha\rightarrow\infty}\frac{E_{\alpha}^{V}}{\alpha^{2}}=e(V).\label{eq:1-6}
\end{equation}
The convergence in (\ref{eq:1-6}) was first argued by M.D. Donsker
and S.R.S. Varadhan in {[}DV1983{]} using large deviation theory.
In 1997, Lieb and L.E. Thomas gave an alternate, pedestrian proof
of the convergence in (\ref{eq:1-6}) using simple modifications of
the Hamiltonian ({[}LT1997{]}), a philosophy that can be traced back
to the inspiring work of Lieb and K. Yamazaki ({[}LY1958{]}). 

In light of the convergence in (\ref{eq:1-6}) for the ground state
energy, it is now only natural to investigate how well Pekar's theory
describes the ground-state wave function (in the strong-coupling limit).
Using the now-standard techniques developed by F. Hiroshima {[}Ha2000{]}
and by M. Griesemer, E.H. Lieb and M. Loss {[}GLL2001{]} to study
the analogous Pauli-Fierz model in quantum electrodynamics, it can
be argued that, under physically natural conditions on the external
potential, the Fröhlich Hamiltonian $H_{\alpha}^{V}$ has a unique
ground state for all values of the coupling paramater $\alpha>0$.
Because it is straightforward to adapt the arguments in {[}Ha2000{]}
and {[}GLL2001{]} to the Fröhlich Hamiltonian and because the arguments
are rather long, we do not provide a proof of the existence and uniqueness
of a ground state here; a sketch of the main ideas is given in the
Appendix. 

Let $\|\Psi_{\alpha}^{V}\|_{\mathcal{F}}^{2}(x)$ denote the electron
density of the ground state, and recall that a minimizer of the Pekar
functional from (\ref{eq:1-4}) is the electronic wave function in
his Produkt-Ansatz. Since the ground state energy in the strong-coupling
limit can be obtained (to a leading order in the electron-phonon coupling)
by minimizing the Pekar functional, shouldn't the electron density
$\|\Psi_{\alpha}^{V}\|_{\mathcal{F}}^{2}$ also converge to a minimizer
of the Pekar functional? Indeed, if the minimization problem in (\ref{eq:1-4})
for the Pekar energy admits a unique minimizer $u_{V}$, then for
all $W\in C_{c}^{\infty}\left(\mathbb{R}^{3}\right)$

\begin{equation}
\lim_{\alpha\rightarrow\infty}\frac{1}{\alpha^{3}}\int_{\mathbb{R}^{3}}\|\Psi_{\alpha}^{V}\|_{\mathcal{F}}^{2}\left(\frac{x}{\alpha}\right)W(x)\,dx=\int_{\mathbb{R}^{3}}\left|u_{V}(x)\right|^{2}W(x)\,dx\label{eq:2.1-1}
\end{equation}
This follows from a technique developed by Lieb and Simon in 1977
(for studying the Thomas-Fermi problem), and consists of differentiating
the (concave) map $\delta\mapsto e(V+\delta W)$ at $\delta=0$, where
\[
e(V+\delta W)=\inf_{\|\psi\|_{2}=1}\left\{ \mathcal{E}_{V}(\psi)-\delta\int_{\mathbb{R}^{3}}\left|\psi(x)\right|^{2}W(x)\,dx\right\} .
\]
However, it is not necessarily the case that the Pekar minimization
problem admits a unique minimizer (See Theorem 1 below). The contribution
of this paper is to address the discrepancy between a unique ground
state and the non-unique Pekar minimizers. 

Let the potential $V_{R}$ be as above. For each $\alpha>0$ the Hamiltonian
$H_{\alpha}^{V_{R}}$, $R>2$ has a unique ground-state wave function.
Since the potential $V_{R}(x)\geq0$ is short-range, i.e. decays exponentially
at infinity, it is known that for each $\alpha>0$ the Schrödigner
operator $\mathbf{p}^{2}-\alpha^{2}V_{R}(\alpha x)$ has a negative
energy bound state in $L^{2}\left(\mathbb{R}^{3}\right)$ (see e.g.
the introduction in {[}BV2004{]}). (To be precise: For the short-range
potential $V_{R}(x)$ it can be seen that there exists \textit{for
all} $\alpha>0$ some $\lambda_{0}>0$ such that for $\lambda>\lambda_{0}$
the operator $\mathbf{p}^{2}-\lambda\alpha^{2}V(\alpha x)$ has a
negative energy bound state in $L^{2}\left(\mathbb{R}^{3}\right)$.
But our proofs still hold true if for some $\lambda>\lambda_{0}$
the function $V_{R}(x)$ in (\ref{eq:P.1}) is replaced by $\lambda V_{R}(x)$,
so we do not inconvenience ourselves any further with this innocuous
technicality.) So, $V_{R}(x)$ satisfies the hypothesis of Proposition
7 in the Appendix. Furthermore, since $V_{R}(x)\geq0$ and $V_{R}\in L^{\infty}\left(\mathbb{R}^{3}\right)$,
the form bound in (\ref{eq:I.5}) follows trivially from Hölder's
inequality; the potential $V_{R}(x)$ also satisfies the hypothesis
of Proposition 8 in the Appendix and the semigroup generated by the
Hamiltonian is positivity improving in the Schroedinger representation.
Hence, for $R>2$ there exists a unique ground-state wave function
$\Psi_{\alpha}^{V_{R}}$, which is therefore invariant after a rotation
in \textit{both} the electron and phonon coordinates. We state this
precisely: Denoting $\hat{\mathbf{n}}$ to be a vector in $\mathbb{R}^{3}$,
the field (phonon) angular momentum relative to the origin is given
by the operator (see {[}Sp2004{]})
\[
J_{f}=\int_{\mathbb{R}^{3}}dk\left(k\times i\nabla_{k}\right)a_{k}^{\dagger}a_{k}.
\]
Let $\mathcal{R}_{\theta}\in SO(3)$ be a rotation by an angle $\theta$
about $\hat{\mathbf{n}}$. Since for any vector $\hat{\mathbf{n}}\in\mathbb{R}^{3}$
and all $\theta$, 
\begin{equation}
\Psi_{\alpha}^{V_{R}}\left(x;\,k\right)=e^{-i\theta\hat{\mathbf{n}}\cdot J_{f}}\Psi_{\alpha}^{V_{R}}\left(\mathcal{R}_{_{\theta}}x;\,k\right),\label{eq:I.6}
\end{equation}
we deduce that the electron density $\|\Psi_{\alpha}^{V_{R}}\|_{\mathcal{F}}^{2}(x)$
is radial for all $R$. 

But we show that the minimizers of the corresponding Pekar functional
are not radial, for $R$ large. Since the Pekar functional is, however,
invariant under rotations, this implies that the non-radial minimizer
is also not unique. 
\begin{thm}
For $R$ large, the Pekar problem $e(V_{R})$ admits only nonradial
minimizers. 
\end{thm}

\noindent We shall show Theorem 1 using a proof by contradiction.
Our arguments use in an essential way Lieb's 1977 uniqueness result
{[}Lb1977{]} for the translation-invariant problem. The discrepancy
in ground-state symmetry shows that the expected convergence in (\ref{eq:2.1-1})
of the (radial) ground-state electron density to a minimizer of the
Pekar functional is not possible. However, we have the following:
\begin{thm}
Let $R$ be large enough so that $e(V_{R})$ in (\ref{eq:1-4}) admits
only nonradial minimizers. Let $\Psi_{\alpha}^{V_{R}}\in\mathcal{H}$
be the unique ground-state wave function of the Fröhlich Hamiltonian
$H_{\alpha}^{V_{R}}$ in (\ref{eq:1-3}). If the minimization problem
in (\ref{eq:1-4}) for the Pekar energy admits a minimizer $u_{V_{R}}$
that is unique up to a rotation, then, denoting $\gamma$ to be the
Haar measure on $SO(3)$, 
\begin{equation}
\lim_{\alpha\rightarrow\infty}\frac{1}{\alpha^{3}}\int_{\mathbb{R}^{3}}\|\Psi_{\alpha}^{V}\|_{\mathcal{F}}^{2}\left(\frac{x}{\alpha}\right)W(x)\,dx=\int_{\mathbb{R}^{3}}\left[\int_{SO(3)}\left|u_{V_{R}}(\mathcal{R}x)\right|^{2}d\gamma(\mathcal{R})\right]W(x)\,dx\label{eq:rv}
\end{equation}
for all $W\in L^{3/2}(\mathbb{R}^{3})+L^{\infty}(\mathbb{R}^{3})$.
\end{thm}

We now describe the strategy for proving Theorem 2. To the Hamiltonian
$H_{\alpha}^{V_{R}}$ we add $\delta$ times the rotational average
$\left\langle W\right\rangle $(x) of a test potential $W(x)\in L^{3/2}(\mathbb{R}^{3})+L^{\infty}(\mathbb{R}^{3})$
that is scaled appropriately: 
\begin{equation}
H_{\alpha}^{V_{R}}-\delta\alpha^{2}\left\langle W\right\rangle \left(\alpha x\right),\label{eq:Pert}
\end{equation}
where $\left\langle W\right\rangle =\int_{SO(3)}W(\mathcal{R}x)\,d\gamma(\mathcal{R})$.
Denoting $E_{\alpha}^{V_{R}+\delta\left\langle W\right\rangle }$
to be the ground-state energy of the Hamiltonian in (\ref{eq:Pert}),
it follows from the variational principle that
\begin{align*}
E_{\alpha}^{V_{R}+\delta\left\langle W\right\rangle } & \leq\left\langle \Psi_{\alpha}^{V_{R}},\,H_{\alpha}^{V_{R}}\Psi_{\alpha}^{V_{R}}\right\rangle -\delta\alpha^{2}\int_{\mathbb{R}^{3}}\left\langle W\right\rangle (\alpha x)\|\Psi_{\alpha}^{V_{R}}\|_{\mathcal{F}}^{2}\left(x\right)\,dx\\
 & =E_{\alpha}^{V_{R}}-\frac{\delta}{\alpha}\int_{\mathbb{R}^{3}}\left\langle W\right\rangle (x)\|\Psi_{\alpha}^{V_{R}}\|_{\mathcal{F}}^{2}\left(\frac{x}{\alpha}\right)\,dx.
\end{align*}

For $\delta>0$, by subraction and division
\[
\frac{E_{\alpha}^{V_{R}+\delta\left\langle W\right\rangle }-E_{\alpha}^{V_{R}}}{\delta\alpha^{2}}\leq-\frac{1}{\alpha^{3}}\int_{\mathbb{R}^{3}}\left\langle W\right\rangle (x)\|\Psi_{\alpha}^{V_{R}}\|_{\mathcal{F}}^{2}\left(\frac{x}{\alpha}\right)\,dx.
\]
By (\ref{eq:1-6}), 
\begin{align}
\frac{e\left(V_{R}+\delta\left\langle W\right\rangle \right)-e\left(V_{R}\right)}{\delta} & \leq\liminf_{\alpha\rightarrow\infty}-\frac{1}{\alpha^{3}}\int_{\mathbb{R}^{3}}\left\langle W\right\rangle (x)\|\Psi_{\alpha}^{V_{R}}\|_{\mathcal{F}}^{2}\left(\frac{x}{\alpha}\right)dx\label{eq:i1}\\
 & =\liminf_{\alpha\rightarrow\infty}-\frac{1}{\alpha^{3}}\int_{\mathbb{R}^{3}}W(x)\|\Psi_{\alpha}^{V_{R}}\|_{\mathcal{F}}^{2}\left(\frac{x}{\alpha}\right)dx;\label{eq:i2}
\end{align}
Above, (\ref{eq:i2}) follows from Fubini's theorem and that $\|\Psi_{\alpha}^{V_{R}}\|_{\mathcal{F}}^{2}(x)$
is a radial function. 

When $\delta<0$, the inequality in (\ref{eq:i1}) is merely reversed
with the ``$\liminf$'' replaced by ``$\limsup$''. Hence, Theorem
2 will follow if the map $\delta\mapsto e\left(V_{R}+\delta\left\langle W\right\rangle \right)$
is differentiable at $\delta=0$. Because the minimization problem
for the energy $e(V_{R})$ does not admit a unique minimizer, the
map $\delta\mapsto e(V_{R}+\delta J)$ cannot be differentiable for
every $J\in L^{3/2}(\mathbb{R}^{3})+L^{\infty}(\mathbb{R}^{3})$.
However, since (by assumption) the minimizers $u_{V_{R}}$ for the
energy $e(V_{R})$ are unique up to rotation, we will show that for
all \textit{radial} $Z\in L^{3/2}\left(\mathbb{R}^{3}\right)+L^{\infty}\left(\mathbb{R}^{3}\right)$,
\begin{equation}
\lim_{\delta\rightarrow0}\frac{e\left(V_{R}+\delta Z\right)-e\left(V_{R}\right)}{\delta}=-\int_{\mathbb{R}^{3}}Z(x)\left|u_{V_{R}}(x)\right|^{2}dx.\label{eq:i3}
\end{equation}
Choosing $Z(x)=\left\langle W\right\rangle (x)$ in (\ref{eq:i3}),
Theorem 2 then follows from Fubini's theorem. 

The paper is organized as follows. A proof of Theorem 1 will be given
in Section 2 below. In Section 3, we establish the crucial differentiation
result, (\ref{eq:i3}), and then prove Theorem 2. 

\section{Nonradiality of the Pekar Minimizers}

Let the Pekar functional $\mathcal{E}_{V}$ be as given in (\ref{eq:1-4})
above. We consider the potential $V_{R}\in C_{c}^{\infty}(\mathbb{R}^{3})$,
$0\leq V_{R}\leq1$ given in (\ref{eq:P.1}) above. The corresponding
Pekar problem is 
\begin{equation}
e(V_{R})=\inf\left\{ \mathcal{E}_{V_{R}}\left(\varphi\right):\ \|\varphi\|_{2}=1\right\} .\label{eq:R.0}
\end{equation}
Since $V_{R}$ vanishes at infinity, by Lions' Concentration Compactness
Principle, we have the following: 
\begin{lem}
The minimization problem in (\ref{eq:R.0}) for the energy $e(V_{R})$
admits a minimizer. 
\end{lem}

\begin{proof}
Theorem III.1 in {[}Ls1984{]}. 
\end{proof}
The goal of this section is to show that the minimizers in the above
Lemma for the energy $e(V_{R})$ are nonradial. But first, we consider
the radial minimization problem: 
\begin{lem}
The minimization problem 
\begin{equation}
e^{\text{rad}}\left(V_{R}\right)=\inf\left\{ \mathcal{E}_{V_{R}}\left(\varphi\right):\ \varphi\in H_{\text{rad}}^{1}\left(\mathbb{R}^{3}\right)\ \text{and}\ \|\varphi\|_{2}=1\right\} \label{eq:R.1}
\end{equation}
admits a (radial) minimizer. 
\end{lem}

\noindent We do not provide a proof of the above Lemma, because it
is standard (cf. Remark III.2 in {[}Ls1984{]} and also {[}Ls1981{]})
and proceeds along the lines of the argument from ``Step 3'' in
the proof of Theorem 3 below; the main ingredient is the well-known
observation of W.A. Strauss (``Radial Lemma 1'' in {[}Ss1977{]})
that any $u\in H_{\text{rad}}^{1}\left(\mathbb{R}^{3}\right)$ satisfies
\begin{equation}
\left|u(x)\right|\leq\frac{\sqrt{2}\left|\mathbb{S}^{2}\right|^{-\frac{1}{2}}\|u\|_{H^{1}}}{\left|x\right|}\ \text{for\ a.e. \ensuremath{|x|\geq2.}}\label{eq:R.2}
\end{equation}
Indeed, with $u\in C_{c}^{\infty}\left(\mathbb{R}^{3}\right)\cap H_{\text{rad}}^{1}\left(\mathbb{R}^{3}\right)$
(we abuse notation by writing $u(x)=u(r)$ with $r=\left|x\right|$),
\[
\left(r^{2}u^{2}\right)_{r}=2\left(ru\right)_{r}\left(ru\right)\leq\left(ru\right)_{r}^{2}+\left(ru\right)^{2}=r^{2}\left(u_{r}^{2}+u^{2}\right)+\left(ru^{2}\right)_{r}.
\]
Then for all $L\geq2$, 
\[
\frac{u^{2}\left(L\right)}{2}L^{2}\leq u^{2}(L)\left(L^{2}-L\right)\leq\int_{0}^{L}\left(u_{r}^{2}+u^{2}\right)r^{2}\,dr\leq\left|\mathbb{S}^{2}\right|^{-1}\left\Vert u\right\Vert _{_{H^{1}}}^{2},
\]
and (\ref{eq:R.2}) follows from a density argument. 

The minimizers for $e^{\text{rad}}\left(V_{R}\right)$ from the above
Lemma play an important role in our proof of nonradiality: 
\begin{lem}
Let the potential $V_{R}$ be as given in (\ref{eq:P.1}), and let
the energies $e\left(V_{R}\right)$ and $e^{\text{rad}}\left(V_{R}\right)$
be as defined by the minimization problems in (\ref{eq:R.0}) and
(\ref{eq:R.1}) respectively. For $R$ large, 
\begin{equation}
e\left(V_{R}\right)<e^{\text{rad}}\left(V_{R}\right).\label{eq:N.1}
\end{equation}
\end{lem}

\begin{proof}
Essential to the proof is the \textit{Free Pekar Problem} (i.e. without
an external potential): 
\begin{equation}
e(0)=\inf_{\|\psi\|_{2}=1}\mathcal{E}_{0}\left(\psi\right)\label{eq:N.2}
\end{equation}
where
\begin{equation}
\mathcal{E}_{0}\left(\psi\right)=\int_{\mathbb{R}^{3}}|\nabla\psi|^{2}dx-\int\int_{\mathbb{R}^{3}\times\mathbb{R}^{3}}\frac{\left|\psi(x)\right|^{2}\left|\psi(y)\right|^{2}}{\left|x-y\right|}dx\,dy.\label{eq:N.3}
\end{equation}

We recall that the problem in (\ref{eq:N.2}) admits a symmetric decreasing
minimizer $Q\in H^{1}(\mathbb{R}^{3})$ with $\|Q\|_{2}=1$ (Theorem
7 in {[}Lb1977{]}). We consider the translate 
\begin{equation}
Q_{R}(x):=Q(x-\zeta_{R})\ \text{with}\ \zeta_{R}=\left(\frac{R+2}{2},0,0\right).\label{eq:N.4}
\end{equation}
Since the functional in (\ref{eq:N.3}) is invariant under translations,
\[
\mathcal{E}_{0}\left(Q_{R}\right)=\mathcal{E}_{0}\left(Q\right)=e(0)\ \text{for all}\ R.
\]
Most importantly, for $R$ large the nonradial function $Q_{R}$ is
concentrated in the potential well of $V_{R}$ located at $\left\{ 2\leq\left|x\right|\le R\right\} $:
Indeed, for $R>2$
\[
\int_{\left\{ 2\leq|x|\leq R\right\} }\left|Q_{R}(x)\right|^{2}dx\geq\int_{\left\{ \left|x-\zeta_{R}\right|\leq\frac{R-2}{2}\right\} }\left|Q_{R}(x)\right|^{2}dx=\int_{\left\{ \left|x\right|\leq\frac{R-2}{2}\right\} }\left|Q(x)\right|^{2}dx
\]
and $\|Q\|_{2}=1$, so 
\begin{equation}
\lim_{R\rightarrow\infty}\int_{\left\{ 2\leq\left|x\right|\leq R\right\} }\left|Q_{R}(x)\right|^{2}dx=1.\label{eq:N.5}
\end{equation}

\noindent \textbf{\bigskip{}
}

\noindent \textbf{Step 1 }(Variational Principle). For all $R$, by
the variational principle, 
\begin{equation}
e(V_{R})\leq\mathcal{E}_{V_{R}}\left(Q_{R}\right)=e(0)-\int_{\mathbb{R}^{3}}V_{R}(x)\left|Q_{R}(x)\right|^{2}dx.\label{eq:N.6}
\end{equation}
By the above Lemma, there is a radial function $\rho_{R}\in H^{1}\left(\mathbb{R}^{3}\right)$
with $\|\rho_{R}\|_{2}=1$ and $\mathcal{E}_{V_{R}}\left(\rho_{R}\right)=e^{\text{rad}}\left(V_{R}\right).$
Hence the claimed inequality in () will follow if we can prove for
$R$ large, 
\begin{equation}
\mathcal{E}_{V_{R}}\left(Q_{R}\right)<\mathcal{E}_{V_{R}}\left(\rho_{R}\right).\label{eq:N.7}
\end{equation}

\noindent \textbf{\bigskip{}
}

\noindent \textbf{Step 2 }(Proof by Contradiction). Suppose (\ref{eq:N.7})
is not true. Then there is a sequence $\{R_{n}\}_{n=1}^{\infty}$
where $R_{n}\rightarrow\infty$ as $n\rightarrow\infty$ and $\mathcal{E}_{V_{R_{n}}}\left(Q_{R_{n}}\right)\geq\mathcal{E}_{V_{R_{n}}}\left(\rho_{R_{n}}\right)$,
i.e. 
\[
e(0)-\int_{\mathbb{R}^{3}}V_{R_{n}}(x)\left|Q_{R_{n}}(x)\right|^{2}dx\geq\mathcal{E}_{0}\left(\rho_{R_{n}}\right)-\int_{\mathbb{R}^{3}}V_{R_{n}}(x)\left|\rho_{R_{n}}(x)\right|^{2}dx.
\]
Then, since $e(0)\leq\mathcal{E}_{0}(\rho_{R_{n}})$, 
\begin{equation}
0\leq\mathcal{E}_{0}\left(\rho_{R_{n}}\right)-e(0)\leq\int_{\mathbb{R}^{3}}V_{R_{n}}(x)\left|\rho_{R_{n}}(x)\right|^{2}dx-\int_{\mathbb{R}^{3}}V_{R_{n}}(x)\left|Q_{R_{n}}(x)\right|^{2}dx.\label{eq:N.8}
\end{equation}

We recall that $0\leq V_{R_{n}}(x)\leq1$ and $V_{R_{n}}(x)=1$ when
$2\leq\left|x\right|\leq R_{n}$. By Hölder's inequality, 
\[
\int_{\left\{ 2\leq\left|x\right|\leq R_{n}\right\} }\left|Q_{R_{n}}(x)\right|^{2}dx\leq\int_{\mathbb{R}^{3}}V_{R_{n}}(x)\left|Q_{R_{n}}(x)\right|^{2}dx\leq1.
\]
Then, by our observation in (\ref{eq:N.5}), 
\begin{equation}
\lim_{n\rightarrow\infty}\int_{\mathbb{R}^{3}}V_{R_{n}}(x)\left|Q_{R_{n}}(x)\right|^{2}dx=1.\label{eq:N.9}
\end{equation}
Furthermore, by the inequalities in (\ref{eq:N.8}) and Hölder's inequality,
\[
\int_{\mathbb{R}^{3}}V_{R_{n}}(x)\left|Q_{R_{n}}(x)\right|^{2}dx\leq\int_{\mathbb{R}^{3}}V_{R_{n}}(x)\left|\rho_{R_{n}}(x)\right|^{2}dx\leq1.
\]
We conclude from (\ref{eq:N.9}) that 
\begin{equation}
\lim_{n\rightarrow\infty}\int_{\mathbb{R}^{3}}V_{R_{n}}(x)\left|\rho_{R_{n}}(x)\right|^{2}dx=1.\label{eq:N.10}
\end{equation}

We deduce from (\ref{eq:N.9}), (\ref{eq:N.10}) and the inequalities
in (\ref{eq:N.8}) that 
\begin{equation}
\lim_{n\rightarrow\infty}\mathcal{E}_{0}\left(\rho_{R_{n}}\right)=e(0).\label{eq:N.11}
\end{equation}
Moreover, $\|\rho_{R_{n}}\|_{2}=1$ and $V_{R_{n}}(x)=0$ when $|x|\leq1$,
so 
\[
\int_{\left\{ |x|\leq1\right\} }\left|\rho_{R_{n}}(x)\right|^{2}dx=1-\int_{\left\{ \left|x\right|>1\right\} }\left|\rho_{R_{n}}(x)\right|^{2}dx\leq1-\int_{\mathbb{R}^{3}}V_{R_{n}}(x)\left|\rho_{R_{n}}(x)\right|^{2}dx.
\]
We then conclude from (\ref{eq:N.10}) that 
\begin{equation}
\lim_{n\rightarrow\infty}\int_{\left\{ |x|\leq1\right\} }\left|\rho_{R_{n}}(x)\right|^{2}dx=0.\label{eq:N.12}
\end{equation}
\textbf{\bigskip{}
}

\noindent \textbf{Step 3} (Conclusion). Recall $\rho_{R}\in H^{1}(\mathbb{R}^{3}),$
$\|\rho_{R}\|_{2}=1$ and $\mathcal{E}_{V_{R}}\left(\rho_{R}\right)=e^{\text{rad}}\left(V_{R}\right).$
Seeking a contradiction, we have shown (see (\ref{eq:N.11}) and (\ref{eq:N.12}))
that for some $R_{n}\rightarrow\infty$ as $n\rightarrow\infty$,
the sequence of radial functions $\left\{ \rho_{R_{n}}\right\} _{n=1}^{\infty}$
is vanishing on the unit ball while also minimizing for the \textit{Free
Pekar Problem} in (\ref{eq:N.2}). Moreover, we recall a result of
E.H. Lieb (Theorem 10 in {[}Lb1977{]}) that this minimization problem
in (\ref{eq:N.2}) admits a symmetric decreasing minimizer $Q\in H^{1}(\mathbb{R}^{3}),$
which is \textit{unique up to translation}. 

Since $\left\{ \rho_{R_{n}}\right\} _{n=1}^{\infty}$ is minimizing
for the problem in (\ref{eq:N.2}), by a standard argument ({[}Lb1977{]})
using Sobolev's and Young's inequalities, 
\begin{equation}
\|\rho_{R_{n}}\|_{H^{1}}<C\label{eq:N.13}
\end{equation}
for all $n.$ Then there is a subsequence, which (with an abuse of
notation) we also denote by $\left\{ \rho_{R_{n}}\right\} _{n=1}^{\infty}$,
and some $\rho\in H^{1}(\mathbb{R}^{3})$ where 
\begin{equation}
\rho_{R_{n}}\rightharpoonup\rho\ \text{in}\ H^{1}\left(\mathbb{R}^{3}\right).\label{eq:N.14}
\end{equation}

We tabulate some immediate observations about $\rho$: $\left\{ \rho_{R_{n}}\right\} _{n=1}^{\infty}$
is radial, so the weak limit $\rho$ is radial almost everywhere.
Moreover, by the weak lower semicontinuity of the $L^{2}$-norm, 
\begin{equation}
\|\rho\|_{2}\leq\liminf_{n\rightarrow\infty}\left\Vert \rho_{R_{n}}\right\Vert _{2}=1\label{eq:N.15}
\end{equation}
and 
\begin{equation}
\left\Vert \nabla\rho\right\Vert _{2}\leq\liminf_{n\rightarrow\infty}\left\Vert \nabla\rho_{R_{n}}\right\Vert _{2}.\label{eq:N.16}
\end{equation}
Finally, since the subsequence $\left\{ \rho_{R_{n}}\right\} _{n=1}^{\infty}$
vanishes on the unit ball (see (\ref{eq:N.12})), by the Rellich-Kondrashov
theorem (Theorem 8.6 in {[}LL2001{]}), 
\begin{equation}
\int_{\left\{ |x|\leq1\right\} }\left|\rho(x)\right|^{2}dx=\lim_{n\rightarrow\infty}\int_{\left\{ |x|\leq1\right\} }\left|\rho_{R_{n}}(x)\right|^{2}dx=0.\label{eq:N.17}
\end{equation}

We shall argue that this weak limit $\rho$\textendash{} an a.e. radial
function vanishing on the unit ball (see (\ref{eq:N.17}))\textendash is
in fact a minimizer for the Free Pekar Problem in (\ref{eq:N.2});
appealing to Lieb's uniqueness result, we then have a contradiction.
The main task is to show 
\begin{equation}
\int\int_{\mathbb{R}^{3}\times\mathbb{R}^{3}}\frac{\left|\rho_{R_{n}}(x)\right|^{2}\left|\rho_{R_{n}}(y)\right|^{2}}{\left|x-y\right|}dx\,dy\longrightarrow\int\int_{\mathbb{R}^{3}\times\mathbb{R}^{3}}\frac{\left|\rho(x)\right|^{2}\left|\rho(y)\right|^{2}}{\left|x-y\right|}dx\,dy.\label{eq:N.18}
\end{equation}
From the positivity of the Coulomb energy (Theorem 9.8 in {[}LL2001{]}),
\[
\left|\left(\int\int_{\mathbb{R}^{3}\times\mathbb{R}^{3}}\frac{\left|\rho_{R_{n}}(x)\right|^{2}\left|\rho_{R_{n}}(y)\right|^{2}}{\left|x-y\right|}dx\,dy\right)^{\frac{1}{2}}-\left(\int\int_{\mathbb{R}^{3}\times\mathbb{R}^{3}}\frac{\left|\rho(x)\right|^{2}\left|\rho(y)\right|^{2}}{\left|x-y\right|}dx\,dy\right)^{\frac{1}{2}}\right|
\]
\begin{equation}
\leq\left(\int\int_{\mathbb{R}^{3}\times\mathbb{R}^{3}}\frac{\left|\left(\rho_{R_{n}}-\rho\right)(x)\right|^{2}\left|\left(\rho_{R_{n}}-\rho\right)(y)\right|^{2}}{\left|x-y\right|}dx\,dy\right)^{\frac{1}{2}}\label{eq:N.19}
\end{equation}
Since $\rho_{R_{n}},\rho$ are radial (we abuse notation by writing
$\rho_{R_{n}}(r)=\rho_{R_{n}}(x)$ with $r=|x|$), by Newton's Theorem
(Theorem 9.7 in {[}LL2001{]}), 
\[
\int\int_{\mathbb{R}^{3}\times\mathbb{R}^{3}}\frac{\left|\left(\rho_{R_{n}}-\rho\right)(x)\right|^{2}\left|\left(\rho_{R_{n}}-\rho\right)(y)\right|^{2}}{\left|x-y\right|}dx\,dy
\]
\[
=\left(4\pi\right)^{2}\int_{0}^{\infty}\left|\left(\rho_{R_{n}}-\rho\right)(s)\right|^{2}\left(\left|\left(\rho_{R_{n}}-\rho\right)(r)\right|^{2}\min\left(r^{-1},s^{-1}\right)r^{2}dr\right)s^{2}ds
\]
\[
\leq\left(4\pi\right)^{2}\left(\int_{0}^{\infty}\left|\left(\rho_{R_{n}}-\rho\right)(s)\right|^{2}s^{2}ds\right)\left(\int_{0}^{\infty}\frac{\left|\left(\rho_{R_{n}}-\rho\right)(r)\right|^{2}}{r}\,r^{2}dr\right)
\]
\begin{equation}
\leq16\pi\left(\int_{0}^{\infty}\frac{\left|\left(\rho_{R_{n}}-\rho\right)(r)\right|^{2}}{r}\,r^{2}dr\right)\label{eq:N.20}
\end{equation}
From Strauss' Radial Lemma ({[}Ss1977{]}; see also (\ref{eq:R.2}))
and the bounds in (\ref{eq:N.13}), (\ref{eq:N.15}) and (\ref{eq:N.16}),
\begin{equation}
\left|\left(\rho_{R_{n}}-\rho\right)(r)\right|\leq\frac{\sqrt{2}\left|\mathbb{S}^{2}\right|^{-\frac{1}{2}}\|\left(\rho_{R_{n}}-\rho\right)(r)\|_{H^{1}}}{r}<\frac{C}{r}\ \text{when}\ r>2.\label{eq:N.21}
\end{equation}
Denoting $B_{M}(0)$ to be a ball of radius $M$ centered at the origin,
by (\ref{eq:N.21}) and Hölder's inequality, 
\[
\int_{0}^{\infty}\frac{\left|\left(\rho_{R_{n}}-\rho\right)(r)\right|^{2}}{r}\,r^{2}dr\leq\frac{\sqrt{M}}{2\sqrt{\pi}}\left\Vert \rho_{R_{n}}-\rho\right\Vert _{L^{4}\left(B_{M}(0)\right)}^{2}+\frac{\left\Vert \rho_{R_{n}}-\rho\right\Vert _{2}}{2\sqrt{\pi}}\left(\int_{M}^{\infty}\frac{\left|\left(\rho_{R_{n}}-\rho\right)(r)\right|^{2}}{r}\,r^{2}dr\right)^{\frac{1}{2}}
\]
\[
\leq\frac{\sqrt{M}}{2\sqrt{\pi}}\left\Vert \rho_{R_{n}}-\rho\right\Vert _{L^{4}\left(B_{M}(0)\right)}^{2}+\frac{C\left\Vert \rho_{R_{n}}-\rho\right\Vert _{2}}{2\sqrt{\pi}}\left(\int_{M}^{\infty}\frac{1}{r^{2}}\,dr\right)^{\frac{1}{2}}
\]
\begin{equation}
\leq\frac{\sqrt{M}}{2\sqrt{\pi}}\left\Vert \rho_{R_{n}}-\rho\right\Vert _{L^{4}\left(B_{M}(0)\right)}^{2}+\frac{C}{\sqrt{\pi M}}.\label{eq:N.22}
\end{equation}
Above, $M$ can be chosen arbitrarily large. Furthermore, by (\ref{eq:N.14})
and the Rellich-Kondrashov theorem (Theorem 8.6 in {[}LL2001{]}),
$\left\Vert \rho_{R_{n}}-\rho\right\Vert _{L^{4}\left(B_{M}(0)\right)}\rightarrow0$
for all $M$. Therefore, the desired convergence in (\ref{eq:N.18})
follows from (\ref{eq:N.19}), (\ref{eq:N.20}) and (\ref{eq:N.22}). 

Since $\left\Vert \rho\right\Vert _{_{2}}\leq1$ (see (\ref{eq:N.15}))
and $e(0)<0$ (Lemma 1(i) in {[}Lb1977{]}), 
\begin{equation}
\mathcal{E}_{0}\left(\rho\right)\geq\mathcal{E}_{0}\left(\frac{\rho}{\|\rho\|_{2}}\right)\|\rho\|_{2}^{2}\geq e(0)\left\Vert \rho\right\Vert _{2}^{2}\geq e(0).\label{eq:N.23}
\end{equation}
Also, by (\ref{eq:N.11}), (\ref{eq:N.16}) and (\ref{eq:N.18}), 

\begin{equation}
e(0)=\lim_{k\rightarrow\infty}\mathcal{E}_{0}\left(\rho_{R_{n}}\right)\geq\mathcal{E}_{0}\left(\rho\right).\label{eq:N.24}
\end{equation}
We deduce from (\ref{eq:N.23}) and (\ref{eq:N.24}) that the weak
limit $\rho$ is a minimizer for the Free Pekar Problem in (\ref{eq:N.2}):
\begin{equation}
\|\rho\|_{_{2}}=1\ \ \ \text{and}\ \ \ \mathcal{E}_{0}\left(\rho\right)=e(0).\label{eq:N.25}
\end{equation}

By Lieb's uniqueness result (Theorem 10 in {[}Lb1977{]}), 
\begin{equation}
\rho(x)=Q(x-a)\ \ \text{for some}\ \ a\in\mathbb{R}^{3},\label{eq:N.26}
\end{equation}
where $Q$ is symmetric decreasing about the origin. But $\rho$ is
a.e. radial, so $a=0$ in (\ref{eq:N.26}) necessarily. Alas, the
minimizer $\rho$ with $\|\rho\|_{2}=1$ is symmetric decreasing about
the origin, and yet $\rho(x)=0$ for a.e. $|x|\leq1$ (see (\ref{eq:N.17}));
we have a contradiction. 
\end{proof}
Theorem 1 now follows. 

\section{The Rotational Average}

As explained in the introduction, we first need to differentiate the
the map $\delta\mapsto e(V+\delta Z)$ for radial test potentials
$Z\in L^{3/2}\left(\mathbb{R}^{3}\right)+L^{\infty}(\mathbb{R}^{3})$. 
\begin{thm}
Let the potential $V\in L^{3/2}\left(\mathbb{R}^{3}\right)+L^{\infty}(\mathbb{R}^{3})$
be nonnegative, vanishing at infinity and not almost everywhere identically-zero.
For a function $W\in L^{3/2}(\mathbb{R}^{3})+L^{\infty}(\mathbb{R}^{3})$
and a real parameter $\delta,$ consider the perturbed Pekar energy
\begin{equation}
e(V+\delta W)=\inf_{\|u\|_{2}=1}\mathcal{E}_{V+\delta W}(u):=\inf_{\|u\|_{2}=1}\left\{ \mathcal{E}_{V}(u)-\delta\int_{\mathbb{R}}W(x)|u(x)|^{2}dx\right\} ,\label{eq:D1}
\end{equation}
where 
\[
\mathcal{E}_{V}(u)=\int_{\mathbb{R}^{3}}|\nabla u|^{2}dx-\int\int\frac{|u(x)|^{2}|u(y)|^{2}}{|x-y|}dx\,dy-\int_{\mathbb{R}^{3}}V(x)|u(x)|^{2}dx.
\]
If the minimization problem for the Pekar energy $e(V)=\inf\{\mathcal{E}_{V}(u):\ \|u\|_{2}=1\}$
admits a minimizer $u_{V}$ that is unique up to rotations, then for
all radial functions $Z\in L^{3/2}(\mathbb{R}^{3})+L^{\infty}(\mathbb{R}^{3})$
the map $\delta\mapsto e(V+\delta Z)$ is differentiable at $\delta=0$
and 
\begin{equation}
\left.\frac{d}{d\delta}\right|_{\delta=0}e(V+\delta Z)=-\int_{\mathbb{R}^{3}}Z(x)|u_{V}(x)|^{2}dx.\label{eq:D2}
\end{equation}
\end{thm}

\begin{proof}
For $W\in L^{3/2}(\mathbb{R}^{3})+L^{\infty}(\mathbb{R}^{3})$, by
a standard argument ({[}Lb1977{]}, {[}LL2001{]}) using Sobolev's and
Young's inequalities, there exist constants $0<c_{1}<1$ and $c_{2}>0$
such that for all $u\in H^{1}(\mathbb{R}^{3})$ with $\|u\|_{2}=1$
and $|\delta|$ sufficiently small, 
\begin{equation}
\mathcal{E}_{V+\delta W}(u)\geq c_{1}\|\nabla u\|_{2}^{2}-c_{2}.\label{eq:D3}
\end{equation}
Therefore, 
\begin{equation}
e(V+\delta W)>-\infty.\label{eq:D4}
\end{equation}

We deduce from (\ref{eq:D4}) that for $W\in L^{3/2}(\mathbb{R}^{3})+L^{\infty}(\mathbb{R}^{3})$
(and $|\delta|$ sufficiently small), the perturbed problem in (\ref{eq:D1})
admits an approximate minimizer $u_{\delta}\in H^{1}(\mathbb{R}^{3})$
with $\|u_{\delta}\|_{2}=1$ satisfying 
\begin{equation}
\mathcal{E}_{V+\delta W}(u_{\delta})\leq e(V+\delta W)+\delta^{2}.\label{eq:D5}
\end{equation}

We denote the set of minimizers for the Pekar energy as $\mathcal{M}:=\{u\in H^{1}(\mathbb{R}^{3}):\ \|u\|_{2}=1\ \text{and}\ \mathcal{E}_{V}(u)=e(V)\}.$
For any $\tilde{u}\in\mathcal{M}$, by the variational principle,
\begin{equation}
e(V+\delta W)\leq\mathcal{E}_{V+\delta W}(\tilde{u})=e(V)-\delta\int_{\mathbb{R}^{3}}W(x)|\tilde{u}(x)|^{2}dx.\label{eq:D5.1}
\end{equation}
Likewise, for an approximate minimizer $u_{\delta},$ $\|u_{\delta}\|_{2}=1$
satisfying (\ref{eq:D5}), 
\[
e(V)\leq\mathcal{E}_{V}(u_{\delta})=\mathcal{E}_{V+\delta W}(u_{\delta})+\delta\int_{\mathbb{R}^{3}}W(x)|u_{\delta}(x)|^{2}dx
\]
\begin{equation}
\leq e(V+\delta W)+\delta^{2}+\delta\int_{\mathbb{R}^{3}}W(x)|u_{\delta}(x)|^{2}dx.\label{eq:D5.2}
\end{equation}

Let $\delta>0$. For a perturbation $W\in L^{3/2}(\mathbb{R}^{3})+L^{\infty}(\mathbb{R}^{3})$
and an approximate minimizer $u_{\delta},$ $\|u_{\delta}\|_{2}=1$
in (\ref{eq:D5}), by the inequalities in (\ref{eq:D5.1}) and (\ref{eq:D5.2}),
\begin{equation}
-\int_{\mathbb{R}^{3}}W(x)|u_{\delta}(x)|^{2}dx-\delta\leq\frac{e(V+\delta W)-e(V)}{\delta}\leq-\left(\sup_{u\in\mathcal{M}}\int_{\mathbb{R}^{3}}W(x)|u(x)|^{2}dx\right).\label{eq:D6}
\end{equation}
When $\delta<0$, the inequalities in (\ref{eq:D6}) are merely reversed:
\begin{equation}
-\int_{\mathbb{R}^{3}}W(x)|u_{\delta}(x)|^{2}dx-\delta\geq\frac{e(V+\delta W)-e(V)}{\delta}\geq-\left(\inf_{u\in\mathcal{M}}\int_{\mathbb{R}^{3}}W(x)|u(x)|^{2}dx\right).\label{eq:D7}
\end{equation}

By our uniqueness assumption, $\mathcal{M}=\{u_{V}(\mathcal{R}x):\ \mathcal{R}\in SO(3)\}$.
Furthermore, with radial functions $Z\in L^{3/2}(\mathbb{R}^{3})+L^{\infty}(\mathbb{R}^{3}),$
by a change of variable, 
\begin{equation}
\int_{\mathbb{R}^{3}}Z(x)|u_{V}(\mathcal{R}x)|^{2}dx=\int_{\mathbb{R}^{3}}Z(x)|u_{V}(x)|^{2}dx\label{eq:D8}
\end{equation}
for all $\mathcal{R}\in SO(3)$. Then, for radial perturbations, the
rightmost quantities in the inequalities (\ref{eq:D6}) and (\ref{eq:D7})
are equal. Hence (with $Z$ radial) the claimed differentiability
of the map $\delta\mapsto e(V+\delta W)$ at $\delta=0$ will follow
from our observation in (\ref{eq:D8}) and the inequalities in (\ref{eq:D6})
and (\ref{eq:D7}) if we can prove the convergence result stated below: 

For radial $Z\in L^{3/2}(\mathbb{R}^{3})+L^{\infty}(\mathbb{R}^{3})$,
let $u_{\delta}$ with $\|u_{\delta}\|_{2}=1$ be an approximate minimizer
as defined in (\ref{eq:D5}) above for the perturbed energy $e(V+\delta Z).$
Then, for any sequence $\{\delta_{n}\}_{n=1}^{\infty}$ where $|\delta_{n}|>0$
and $\delta_{n}\rightarrow0$ as $n\rightarrow\infty$, the corresponding
sequence of approximate minimizer $\{u_{\delta_{n}}\}_{n=1}^{\infty}$
satisfies 
\begin{equation}
\lim_{n\rightarrow\infty}\int_{\mathbb{R}^{3}}Z(x)|u_{\delta_{n}}(x)|^{2}dx=\int_{\mathbb{R}^{3}}Z(x)|u_{V}(x)|^{2}dx.\label{eq:D9}
\end{equation}
We observe that $\{u_{\delta_{n}}\}_{n=1}^{\infty}$ is minimizing
for the problem $e(V)=\inf\{\mathcal{E}_{V}(u):\ \|u\|_{2}=1\}$.
Then, by Lions' concentration compactness argument, every subsequence
$\{u_{\delta_{n_{k}}}\}$ has a sub-subsequence $\{u_{\delta_{n_{k_{l}}}}\}$
converging strongly in $H^{1}(\mathbb{R}^{3})$ to some function in
$\mathcal{M}=\{u_{V}(\mathcal{R}x):\ \mathcal{R\in S}O(3)\}$. We
deduce from our observation in (\ref{eq:D8}) that 
\[
\lim_{l\rightarrow\infty}\int_{\mathbb{R}^{3}}Z(x)|u_{\delta_{n_{k_{l}}}}(x)|^{2}dx=\int_{\mathbb{R}^{3}}Z(x)|u_{V}(\mathcal{R}x)|^{2}dx
\]
\[
=\int Z(x)|u_{V}(x)|^{2}dx
\]
\end{proof}
We are now ready to prove Theorem 2. 
\begin{proof}[Proof of Theorem 2.]
 For any $W\in L^{3/2}(\mathbb{R}^{3})+L^{\infty}(\mathbb{R}^{3})$
we denote its rotational average $\left\langle W\right\rangle =\int_{SO(3)}W(\mathcal{R}x)\,d\gamma(\mathcal{R})$.
Note that $\left\langle W\right\rangle \in L^{3/2}(\mathbb{R}^{3})+L^{\infty}(\mathbb{R}^{3})$.
As explained at the end of the introduction, using the variational
principle and (\ref{eq:1-6}), we arrive at the relations 
\[
\frac{e(V_{R}+\delta\left\langle W\right\rangle )-e(V_{R})}{\delta}\leq\liminf_{\alpha\rightarrow\infty}-\frac{1}{\alpha^{3}}\int_{\mathbb{R}^{3}}\|\Psi_{\alpha}^{V_{R}}\|_{\mathcal{F}}^{2}\left(\frac{x}{\alpha}\right)\,\left\langle W\right\rangle (x)\,dx,
\]
and 
\[
\frac{e(V_{R}+\delta\left\langle W\right\rangle )-e(V_{R})}{\delta}\geq\limsup_{\alpha\rightarrow\infty}-\frac{1}{\alpha^{3}}\int_{\mathbb{R}^{3}}\|\Psi_{\alpha}^{V_{R}}\|_{\mathcal{F}}^{2}\left(\frac{x}{\alpha}\right)\,\left\langle W\right\rangle (x)\,dx.
\]

Using Fubini's theorem, a simple change of variable and that the electron
density $\|\Psi_{\alpha}^{V}\|_{\mathcal{F}}^{2}\left(\frac{x}{\alpha}\right)$
is radial, we observe 
\begin{equation}
\int_{\mathbb{R}^{3}}\|\Psi_{\alpha}^{V_{R}}\|_{\mathcal{F}}^{2}\left(\frac{x}{\alpha}\right)\,\left\langle W\right\rangle (x)\,dx=\int_{\mathbb{R}^{3}}\|\Psi_{\alpha}^{V_{R}}\|_{\mathcal{F}}^{2}\left(\frac{x}{\alpha}\right)W(x)\,dx.\label{12}
\end{equation}
Furthermore, since $\left\langle W\right\rangle $ is radial and we
assume that the problem in (\ref{eq:1-4}) admits a minimizer $u_{V_{R}}$
that is unique up to rotations, we conclude from Theorem 6 and (\ref{12}):
\[
\lim_{\alpha\rightarrow\infty}-\frac{1}{\alpha^{3}}\int_{\mathbb{R}^{3}}\|\Psi_{\alpha}^{V}\|_{\mathcal{F}}^{2}\left(\frac{x}{\alpha}\right)\,W(x)\,dx=\lim_{\alpha\rightarrow\infty}-\frac{1}{\alpha^{3}}\int_{\mathbb{R}^{3}}\|\Psi_{\alpha}^{V}\|_{\mathcal{F}}^{2}\left(\frac{x}{\alpha}\right)\,\left\langle W\right\rangle (x)\,dx\qquad\qquad\qquad\qquad\qquad\qquad\qquad\qquad\qquad\qquad\qquad\qquad\qquad\qquad
\]
\[
\quad=\left.\frac{d}{d\delta}\right|_{\delta=0}e(V+\delta\left\langle W\right\rangle )
\]
\[
\quad\quad\quad=-\int_{\mathbb{R}^{3}}|u_{V_{R}}(x)|^{2}\left\langle W\right\rangle (x)\,dx
\]
\[
\quad\quad\quad\quad\quad\quad\quad\quad\quad\quad=-\int_{\mathbb{R}^{3}}\left(\int_{SO(3)}|u_{V_{R}}(\mathcal{R}x)|^{2}d\gamma(\mathcal{R})\right)W(x)\,dx.
\]
\end{proof}

\section*{Acknowledgement}

I thank Michael Loss for helpful discussions.

\appendix

\section{Existence and Uniqueness of a Ground State}

We describe the main ideas in {[}GLL2001{]} and {[}Ha2000{]} for arguing
the existence of a unique ground-state wave function: 
\begin{prop}
Fix $\alpha>0$. If the Schrödinger operator $\mathbf{p}^{2}-\alpha^{2}V(\alpha x)$
has a negative energy bound state in $L^{2}\left(\mathbb{R}^{3}\right)$,
i.e., there is an eigenfunction $\zeta\in L^{2}\left(\mathbb{R}^{3}\right)$
and $\eta>0$ so that 
\[
\left(\mathbf{p}^{2}-\alpha^{2}V(\alpha x)\right)\zeta(x)=-\eta\zeta(x),
\]
then there exists a normalized function $\Psi_{\alpha}^{V}$ in $L^{2}\left(\mathbb{R}^{3}\right)\otimes\mathcal{F}$
satisfying 
\[
H_{\alpha}^{V}\Psi_{\alpha}^{V}=E^{V}\left(\alpha\right)\Psi_{\alpha}^{V}.
\]
\end{prop}

\noindent The existence of a negative energy bound state of the operator
$\mathbf{p}^{2}-\alpha^{2}V(\alpha x)$ can be used to show that the
Fröhlich Hamiltonian $H_{\alpha}^{V}$ satisfies the binding condition
(cf. Theorem 3.1 in {[}GLL2001{]})
\begin{equation}
E^{V}\left(\alpha\right)<E^{V\equiv0}\left(\alpha\right).\label{eq:I.3}
\end{equation}
With the Rellich-Kondrashov theorem and the binding inequality in
(\ref{eq:I.3}), the above proposition can be established along the
lines of the argument provided in {[}GLL2001{]}. In order to see that
the ground state is unique, we use the well-known Schrödinger representation
of the phonon Fock space $\mathcal{F}$, which is naturally identified
with the $L^{2}$ space over a probability measure space $\left(\mathcal{Q},\mu\right)$
(see p. 185 in {[}Sp2004{]}). We denote the unitary operator 
\begin{equation}
\vartheta:\ L^{2}\left(\mathbb{R}^{3}\right)\otimes\mathcal{F}\mapsto L^{2}\left(\mathbb{R}^{3}\otimes\mathcal{Q},\ dx\times d\mu\right).\label{eq:I.4}
\end{equation}
The identification in (\ref{eq:I.4}) of $\mathcal{F}$ with an $L^{2}$
space opens up the possibility of establishing the uniqueness of the
ground state via the classical route of positiviity improvement: on
a $\sigma-$finite measure space $\left(\chi,\nu\right)$, a bounded
operator $B$ on $L^{2}\left(\chi,\nu\right)$ is positivity improving
if $\langle f_{1},B\,f_{2}\rangle_{L^{2}\left(\chi,\nu\right)}>0$
for all positive $f_{1}$ and $f_{2}$ in $L^{2}\left(\chi,\nu\right)$
(and a function $f\in L^{2}\left(\chi,\nu\right)$ is positive if
$f\geq0$ a.e. and $f\neq0$ a.e.). Armed with the Schrödinger representation
and the notion of positivity improvement defined above, uniqueness
can be shown along the lines of Hiroshima's argument in {[}Ha2000{]}:
\begin{prop}
Fix $\alpha>0$, and let the external potential $V\in L^{3/2}\left(\mathbb{R}^{3}\right)+L^{\infty}\left(\mathbb{R}^{3}\right)$
satisfy the conditions of Proposition 1 above. Writing $V=V_{+}-V_{-}$,
suppose $\alpha^{2}V_{+}(\alpha x)$ is relatively form bounded with
respect to the operator $\mathbf{p}^{2}$ with form bound strictly
less than one; that is, for some $0<a<1$ there exists $c_{a}>0$
such that for all $\xi\in H^{1}\left(\mathbb{R}^{3}\right)$, 
\begin{equation}
\alpha^{2}\int_{\mathbb{R}^{3}}V_{+}\left(\alpha x\right)\left|\xi(x)\right|^{2}dx\leq a\|\nabla\xi\|_{2}^{2}+c_{a}\|\xi\|_{2}^{2}.\label{eq:I.5}
\end{equation}
Then the ground-state wave function $\Psi_{\alpha}^{V}$ of the Fröhlich
Hamiltonian $H_{\alpha}^{V}$ is unique. 
\end{prop}

\noindent Let $\vartheta$ be the unitary operator as given in (\ref{eq:I.4}).
When the external potential $V\in L^{3/2}\left(\mathbb{R}^{3}\right)+L^{\infty}\left(\mathbb{R}^{3}\right)$
satisfies the condition in (\ref{eq:I.5}), it is possible to show
using the functional integral formula for the heat kernel that the
operator $\vartheta e^{-tH_{\alpha}^{V}}\vartheta^{-1}$, $t>0$ is
positivity improving {[}Ha2000{]}. It then follows that the ground
state of $\vartheta H_{\alpha}^{V}\vartheta^{-1}$ is unique (see
p.191 in {[}Sp2004{]}). Since $\vartheta$ is unitary, the ground
state of $H_{\alpha}^{V}$ is therefore also unique.


\begin{thebibliography}{GLL2001}
\bibitem[BV2004]{key-5} Benguria, R.D. and S.A. Vugalter, ``Binding
Threshold for the Pauli-Fierz Operator,'' Letters in Mathematical
Physics \textbf{70}(3), 249-257 (2004). 

\bibitem[Da1989]{key-1} Dacorogna, B., Direct Methods in the Calculus
of Variations, Springer, New York (1989). 

\bibitem[DV1983]{key-8} Donsker, M.D. and S.R.S. Varadhan, ``Asymptotics
of the polaron'', Communications in Pure and Applied Mathematics
\textbf{36}(4), 505-528 (1983). 

\bibitem[FZ1986]{key-5} Fisher, M.P.A and W. Zwerger, ``Ground-state
symmetry of a generalized polaron,'' Physical Review B \textbf{34}(8),
5912-5915 (1986). 

\bibitem[GLL2001]{key-1} Griesemer, M., E.H. Lieb and M. Loss, ``Ground
states in non-relativistic quantum electrodynamics,'' Inventiones
Mathematicae \textbf{145}(3), 557-595 (2001). 

\bibitem[GW2013]{key-1} Griesemer, M. and D. Wellig, ``The strong-coupling
polaron in static electric and magnetic fields,'' Journal of Physics
A: Mathematical and Theoretical \textbf{46}, 425202 (2013)

\bibitem[Ha2000]{key-1} Hiroshima, F., ``Ground states of a model
in nonrelativistic quantum electrodynamics. II,'' Journal of Mathematical
Physics \textbf{41}(2), 661-674 (2000). 

\bibitem[Lo2015]{key-14} Lakhno, V.D., ``Pekar's ansatz and the
strong coupling problem in polaron theory,'' Physics-Uspekhi \textbf{58}
(3), 295-308 (2015). 

\bibitem[Lb1977]{key-2} Lieb, E.H., ``Existence and uniqueness of
the minimizing solution of Choquard's nonlinear equation,'' Studies
in Applied Mathematics \textbf{57}, 93-105 (1977). 

\bibitem[LL2001]{key-2} Lieb, E.H. and M. Loss, \textit{Analysis
}(Second Edition), Graduate Studies in Mathematics \textbf{14}, American
Mathematical Society, Providence, RI (2001). 

\bibitem[LT1997]{key-9} Lieb, E.H. and L.E. Thomas, ``Exact ground
state energy of the strong-coupling polaron,'' Communications in
Mathematical Physics \textbf{183}(3), 511-519 (1997). 

\bibitem[LY1958]{key-10} Lieb, E.H. and K. Yamazaki, ``Ground-State
Energy and Effective Mass of the Polaron,'' Physical Review \textbf{111}(3),
728-733 (1958). 

\bibitem[Ls1981]{key-5} Lions, P.L., ``Minimization problems in
$L^{1}(\mathbb{R}^{3})$,'' Journal of Functional Analysis \textbf{41},
236-275 (1981). 

\bibitem[Ls1984]{key-3} Lions, P.L., ``The concentration-compactness
principle in the Calculus of Variations. The locally compact case,
part I,'' Annales de l'Institut Henri Poincaré (C) Analyse Non Linéaire
\textbf{1}, 109-145 (1984). 

\bibitem[Ny1989]{key-11} Nagy, P., ``A note to the translationally
invariant strong coupling theory of the polaron,'' Czechoslovak Journal
of Physics B \textbf{39} (3), 353-356 (1989). 

\bibitem[Sp2004]{key-3} Spohn, H., \textit{Large Scale Dynamics of
Interacting Particles}, Springer, Berlin (2004). 

\bibitem[Ss1977]{key-6} Strauss, W.A., ``Existence of solitary waves
in higher dimensions,'' Communications in Mathematical Physics \textbf{55},
149-162 (1977). 
\end{thebibliography}
\end{document}